%% file: main.tex
\documentclass[a4paper,UKenglish]{article}
\usepackage{latexsym}
\usepackage{amsmath,amssymb,mathtools}
\usepackage{subcaption}
\usepackage{amsfonts}
\usepackage{authblk}
\usepackage{cite}
\usepackage{wrapfig}
\usepackage{xspace}
\usepackage[linesnumbered,lined,ruled,vlined]{algorithm2e}

\usepackage{authblk}
\usepackage{amsthm}           
\input{macro}

\title{Efficient Enumeration of Induced Matchings in a Graph without Cycles with Length Four}
\author[1]{Kazuhiro Kurita}
\author[2]{Kunihiro Wasa}
\author[2]{Takeaki Uno}
\author[1]{Hiroki Arimura}

\affil[1]{IST, Hokkaido University, Sapporo, Japan\\
  \texttt{\{k-kurita, arim\}@ist.hokudai.ac.jp}}
\affil[2]{National Institute of Informatics, Tokyo, Japan\\
  \texttt{\{wasa, uno\}@nii.ac.jp}}

\begin{document}
\maketitle


\begin{abstract}
  \input{abst}
\end{abstract}
\input{intro}
\input{prelim}
\input{algoc4free}
\input{algomso}
\input{conc}
\input{acknowledge}

\bibliographystyle{abbrv}
\bibliography{kurita}

\newpage
\appendix
\input{appendix}

\input{appendix1}

\end{document}

%% file: macro.tex
\newcommand{\set}[1]{\{#1\}}
\newcommand{\inset}[2]{\{#1 \mid #2\}}
\newcommand{\name}[1]{\textit{#1}}

\newcommand{\size}[2][]{\left| #2^{#1} \right|}
\newcommand{\order}[1]{O\left(#1\right)}
\newtheorem{problem}{Problem}
\newcommand{\argmax}{\mathop{\rm arg~max}\limits}

\newtheorem{lemma}{Lemma}
\newtheorem{theorem}[lemma]{Theorem}
\newtheorem{corollary}[lemma]{Corollary}
\newtheorem{definition}[lemma]{Definition}

\newcommand{\sig}[1]{\mathcal{#1}}
\newcommand{\leid}[1]{<_{ID}}
\newcommand{\geid}[1]{>_{ID}}

\newcommand{\List}[1]{\sig List(#1)}

\makeatletter
\@ifpackageloaded{algorithm2e}{
\SetKwProg{Fn}{}{}{}
\SetKwProg{Procedure}{Procedure}{}{}
\SetKwProg{Subprocedure}{Subprocedure}{}{}
\SetKwFunction{EnumIM}{EIM}
\SetKwFunction{RecEIM}{RecEIM}
\SetKwComment{tcc}{//}{}
\SetKwFunction{Output}{Output}%
\SetKwInOut{AlgInput}{Input}
\SetKwInOut{AlgOutput}{Output}
\SetKwInOut{AlgPrecondition}{Pre-conditions}
\SetKwInOut{AlgInvariant}{Invariants}
}{}
\makeatother

%% file: abst.tex
We address the induced matching enumeration problem. 
An edge set $M$ is an induced matching of a graph $G =(V,E)$. 
The enumeration of matchings are widely studied in literature,
but the induced matching has not been paid much attention.
A straightforward algorithm takes $\order{\size{V}}$ time 
for each solution, that is coming from the time to generate 
a subproblem.
We investigated local structures that enables us to generate 
subproblems in short time,
and proved that the time complexity 
will be $\order{1}$ if the input graph is $C_4$-free.
A $C_4$-free graph is a graph any whose subgraph 
is not a cycle of length four.
Finally, we show the fixed parameter tractability of
counting induced matchings for graphs with
bounded tree-width and planar graphs. 

%% file: intro.tex
\section{Introduction}
\label{subsec:enum}


An \name{enumeration problem} is to output all solutions
to a given problem without duplication.
Enumeration problems and their algorithms
have been continuously studied in literature,
and recently the studies have got more active
from the expansion of 
the applications, such as data mining, network analysis and
computational proofs in mathematics. 
According to the increase of the activity, several problems that solved
in the past were revisited such as paths, cycles, and trees,
and several structures began to be studied~
\cite{
  Edmonds:CJM:1965,
  Ferreira:Grossi:Rizzi:ESA:2011,
  Tarjan:Read:1975,
  Uno:WADS:15
}.
In this paper, we also revisit an old fashioned problem of enumerating 
matchings, but consider its induced version that has not been paid attention much. 

The efficiency of enumeration algorithm is often evaluated by 
\name{output-polynomiality}~%
\cite{Uno:Encyclopedia:2016}.
An algorithm is said to be \name{output-polynomial time} if its total
running time is bounded by $poly(N, M)$, where $N$ is the input size,
$M$ is the output size and $poly(N,M)$ is  a polynomial function on
$N$ and $M$. 
The \name{delay} of an enumeration algorithm is the maximum computation time between 
the output solution and the next solution and time after the
last solution until the termination of the algorithm.
An algorithm is \name{polynomial delay} if its delay is bounded
by polynomial in $N$.
In particular, we say an algorithm is said to runs in $O(poly(N))$ time
for each if the algorithm runs in time linear in $Mpoly(N)$.



In this paper, 
we consider the enumeration problem for induced matchings in the given graph (abbreviated as \name{EIM}).
An induced matching of a graph is an edge set such that the endpoints of any two edges in the set are not adjacent to each other, i.e., the graph induced by the endpoints of the edge set forms a matching.
Uno~\cite{Uno:WADS:15} showed that matchings can be enumerated in a general graph in constant amortized time 
by using amortization technique, called \name{Push out},  
to distribute the cost of each iteration to many descendants. 
We can also consider a straightforward binary partition algorithm
(branch and bound algorithm) for induced matching enumeration
that runs in $\order{\Delta^2}$ time per solution,  
where $\Delta$ is the maximum degree in an input graph. 
The structure of the recursion is quite different from the ordinal 
matching enumeration, thus 
a direct application of the technique described in ~\cite{Uno:WADS:15}
does not work.
The push out technique and the other amortization require some conditions,
but it is not easy to develop algorithms satisfying the conditions.
The existence of more efficient algorithms is still open.


In this kind of low-degree polynomial time enumeration algorithm, 
the most time consuming part is often typically the generation of 
the child problems.
Particularly, we spend much time when the local structure of the 
problem is complicated around the pivot vertex or edge, that is 
to be fixed or to be removed from the problem.
If the structure is simple, the child problem generation can be done
in short time.
For example, if the graph is a tree, there is no cycle around a vertex,
thus we do not have to think about unification of multiple edges when
we shrink an edge.
In this paper, we consider $C_4$-free graphs, and propose 
an algorithm runs in constant time for each, 
where $C_4$-free graphs are graphs that have no cycles of length equal to four.
In an ordinal binary partition algorithm for induced matching enumeration,
we choose an edge $e$ and enumerate induced matchings including $e$.
This is done by enumerating all induced matching included the graph obtained
by removing $e$, edges adjacent to $e$, and edges adjacent to
edges adjacent to $e$.
This takes $\order{\Delta^2}$ time and this is the bottle neck of 
the algorithm.
We investigated the $C_4$-free graphs, and could find that the 
structural property of $C_4$-free graphs makes the process of 
generating the subproblems light.
We introduced new way of branching the problem, 
so that in each iteration we select a vertex $v$ with the maximum degree,
and partition the problem into $\Delta$ subproblems. 
The property together with this branching method lighten the 
computation of subproblems, and the computation time of an iteration
is bounded by the number of its descendants.
This enables us to use amortization analysis, and can obtain the result.

The organization of the paper is as follows.
We show in Sec.~\ref{sec:c4free} that induced matchings are
enumerable in constant amortized time per solution for $C_4$-free graphs.
In Sec.~\ref{secmso}, we show that counting all induced matchings can be solvable in FPT linear time for graphs with bounded degree, bounded tree-width, and planar graphs by using the results of~\cite{frick:TCSyst:2004generalized,arnborg:seese:JOA:1991easy}.
These results seem to show for the first time the complexities of counting and enumeration problems for induced matchings.
\begin{table}[t]
  \caption{Summary of our results and related work.
    This table
    shows the complexity of path, cycle, matching, and
    induced matching for each problems. 
    In each cell, we list the complexities of ordinary problem first and parameterized problems next, 
    where the parameter $k$ is included in input in all non-parameterized problems (*),
    the problem is in FPT for
    graphs with bounded degree, bounded tree-width, girth at least $6$, 
    line graphs, and planar graphs 
    ($^{**}$), 
    the counting problem for $k$-induced matchings is FPT linear for graphs with bounded degree, bounded tree-width, and planar graphs parameterized with an implicit parameter determined by the class ($^{***}$),
    the complexity of an enumeration problem shows its amortized complexity per solution. 
  }   
  \label{tab:related}
  \centering
  \begin{tabular}{lccccccc}
    \hline
    & decision problem & counting problem & enumeration problem \\ \hline \hline
    $k$-path & NP-complete~\cite{Garey:Johnson:2002}$^*$ & \#P-complete~\cite{Valiant:1979:siam}$^*$ & Polynomial~\cite{Read:1975}$^*$ \\
    & FPT~\cite{Fellows:1989} & \#W[1]-complete~\cite{flum2006parameterized} &
    \\ \hline
    $k$-cycle & NP-complete~\cite{Garey:Johnson:2002}$^*$ & \#P-complete~\cite{Valiant:1979:siam}$^*$ & Polynomial~\cite{Ferreira:Grossi:Rizzi:ESA:2011}$^*$ \\
    & FPT~\cite{Downey:2012} & \#W[1]-complete~\cite{flum2006parameterized} &
    \\ \hline
    $k$-matching & P~\cite{Edmonds:CJM:1965}$^*$ & \#P-complete~\cite{Valiant:1979}$^*$ & $\order{1}$~\cite{Uno:WADS:15}$^*$
    \\& & \#W[1]-complete~\cite{Curticapean:2014}
    &
    \\\hline
    \shortstack{
      $k$-induced
      \\matching 
    }
    &
    \raisebox{0.5em}{
      NP-complete~\cite{Stockmeyer:IPL:82}$^*$
    }
    & 
    \raisebox{0.5em}{ Unknown }
    & \strut\raisebox{-0.75em}{
      \shortstack{
      $\order{1}$ for boudned-
      \\degree [Ours]$^*$
    }}
    \\
    &
    \raisebox{0.5em}{
      \shortstack{
        W[1]-complete~\cite{Moser:2009}
        \\FPT$^{**}$~\cite{Moser:2009} 
    }}
    & \shortstack{
      FPT linear$^{***}$ for
      \\bounded tree-width
      \\and planar
            [Ours]} 
    &
    \raisebox{1.0em}{
      \shortstack{
        $\order{1}$ for $C_4$-free [Ours]$^*$
    } }
    \\ \hline
  \end{tabular}
\end{table}

\input{related}

%% file: related.tex
\subsection{Related works}

In Table~\ref{tab:related}, we show the summary of related work and our results
on decision, counting, and enumeration problems for small subgraphs in a graph.
The decision problem for matching (maximum matching, MM) has been extensively
studied for more than 50 years~\cite{Edmonds:CJM:1965,Hopcroft:Karp:SICOMP:1973}.
The decision problem can be solved in polynomial time for matchings 
(See~\cite{Edmonds:CJM:1965,Hopcroft:Karp:SICOMP:1973}), 
while the problem (MIM) for induced matchings is known to be NP-complete~\cite{Stockmeyer:IPL:82}. 
The latter is still NP-hard for 
graphs with bounded degree, 
bipartite graphs, 
$C_4$-free, 
line graphs,
and 
planar graphs
~\cite{Cameron:DAM:1989,Brandstadt:Hoang:2008,Lozin:2002}.
MIM can be solved in polynomial time
for restricted graph classes: 
interval, chordal, weakly chordal,  circular-arc,  trapezoid, and
co-comparability graphs%
~\cite{Brandstadt:Hoang:2008,Golumbic:Lewenstein:DAM:2000}. 

In general, counting of matchings is computational hard.
In particular, the counting of matchings is \#P-complete (Valiant~\cite{Valiant:1979:siam}), while it is  \#W[1]-complete parameterized with the size $k$ of a matching in a bipartite graph (Curticapean and Marx~\cite{Curticapean:2014}).
For enumeration, Uno~\cite{Uno:WADS:15} showed that matchings can be enumerated in constant amortized time per solution.
To the best of our knowledge, there are almost no known results for counting and enumeration of induced matchings.

In this paper, we study the complexity of enumeration problems for graphs without cycles of length~$4$. Since any graph with girth at least $5$ has no $C_4$, 
our algorithm also works for such graphs with large girth, where the girth of a graph is the length of a shortest cycle in the graph. 
Recently, there are a few results showing an interesting interplay between induced matchings and graphs with large girth as follows. 
Raman and Saurabh~\cite{Raman:2008} demonstrated that several fixed parameter intractable problems, such as dominating sets, fall in FPT when input graphs have large girth.
As a most closely related result, Moser and Sikdar~\cite{Moser:2009}  show that 
the $W[1]$-hard decision problem for induced matchings becomes in FPT for graphs with bounded degree and with girth $6$ or more. 

%% file: prelim.tex
 \section{Preliminary}
\label{sec:prelim}

In this paper, for disjoint set $A$ and $B$,
we define disjoint union of $A$ and $B$ by $A \sqcup B$. 
If it is clearly understood,
we denote $V(G) = V$ and $E(G) = E$. 

Let $G = (V, E)$ be an undirected graph with vertex set $V$ and edge set $E \subseteq V\times V$.
In this paper, 
we assume that graphs have no self-loops or parallel edges.
An edge $e$ with vertex $u$ and $v$ is denoted
by $e = \set{u, v}$. 
Two vertices $u, v \in V$ are \name{adjacent} if there is an edge $\set{u, v}$ in $E$.
We say $u$ is a neighbor of $v$ if $u \in N_G(v)$. 
Similarly, two edges $e, f \in E$ are \name{adjacent} 
if $e$ and $f$ share the same vertex. 
Let $N_G(u)$ be the set of neighbors of $u$ in $G$ and   
$N_G[u] = N_G(u) \cup \set{u}$ be the set of closed neighbor of $u$. 
Let $d_G(u) = \size{N_G(u)}$  be the degree of $u$ in $G$. 
$\Delta(G) = \max_{x \in V} d(x)$ denotes the maximum degree of $G$. 
For any vertex subset $V' \subseteq V$, 
we say $G[V'] = (V', E[V'])$ an \name{induced subgraph},
where $E[V'] = \inset{\set{u, v} \in E(G)}{u, v \in V'}$.
Since $G[V']$ is uniquely determined by $V'$, 
we identify $V'$ with $G[V']$. 
We denote by $G \setminus \set{e} = (V, E \setminus \set{e})$ and 
$G\setminus \set{v} = G[V\setminus\set{v}]$.
For simplicity, we denote by $v \in G$ and $e \in G$ if $v \in V$ and $e \in E$, 
respectively.

An alternating sequence $\pi = (v_1, e_1, v_2, \dots, v_{k-1}, e_k, v_k)$ of 
vertices and edges is a \name{path} if 
each edge and vertex in $\pi$ appears at most once. 
We also call $\pi$ an \name{$v_0$-$v_n$ path}. 
Then, 
An alternating sequence $C = (v_1, e_1, v_2, \dots, v_{k-1}, e_k, v_k)$ of 
vertices and edges is a \name{cycle} if 
$(v_1, e_1, v_2, \dots, v_{k-1})$ is a $v_0$-$v_{k-1}$ path and  $v_k = v_1$. 
The length of a path and a cycle is defined by the number of its edges. 
Let $G$ be $C_k$-free graph if $G$ has no cycle with length $k$ as a subgraph. 
For example, if $G$ has no $4$-cycles then $G$ is a $C_4$-free graph. 

For any vertices $u, v \in V$, 
the distance between $u$ and $v$ is defined 
by the length of a shortest $u$-$v$ path. 
The distance between edge $e$ and $f$ is
defined by the length of a shortest path, i.e., 
$dist_G(e, f) = \min\inset{dist_G(u, v)}{u \in e, v \in f}$. 
Similarly, the distance between vertex $v$ and edge $e$ is
defined by $dist_G(v, e) = \min\inset{dist_G(u, v)}{u \in e}$. 

A \name{matching} in a graph $G = (V, E)$ is an edge subset $M \subseteq E$ 
of $G$, if any pair of edges in $M$ does not share their endpoints. 
An \name{induced matching} in a graph $G = (V, E)$ is
a matching $M \subseteq E$ whose vertex set induces $M$ itself.
In other words, an edge set $M$ is an induced matching if and only if 
$dist_G(e, e') \ge 2$ for any distinct edges $e, e' \in M$, i.e.,
there is no edge $f$ in $G$ connecting $e$ and $e'$. 
In figure~\ref{fig:im}, we show an example of an induced matching. 

\begin{figure}[t]
  \centering
  \includegraphics[width=100mm]{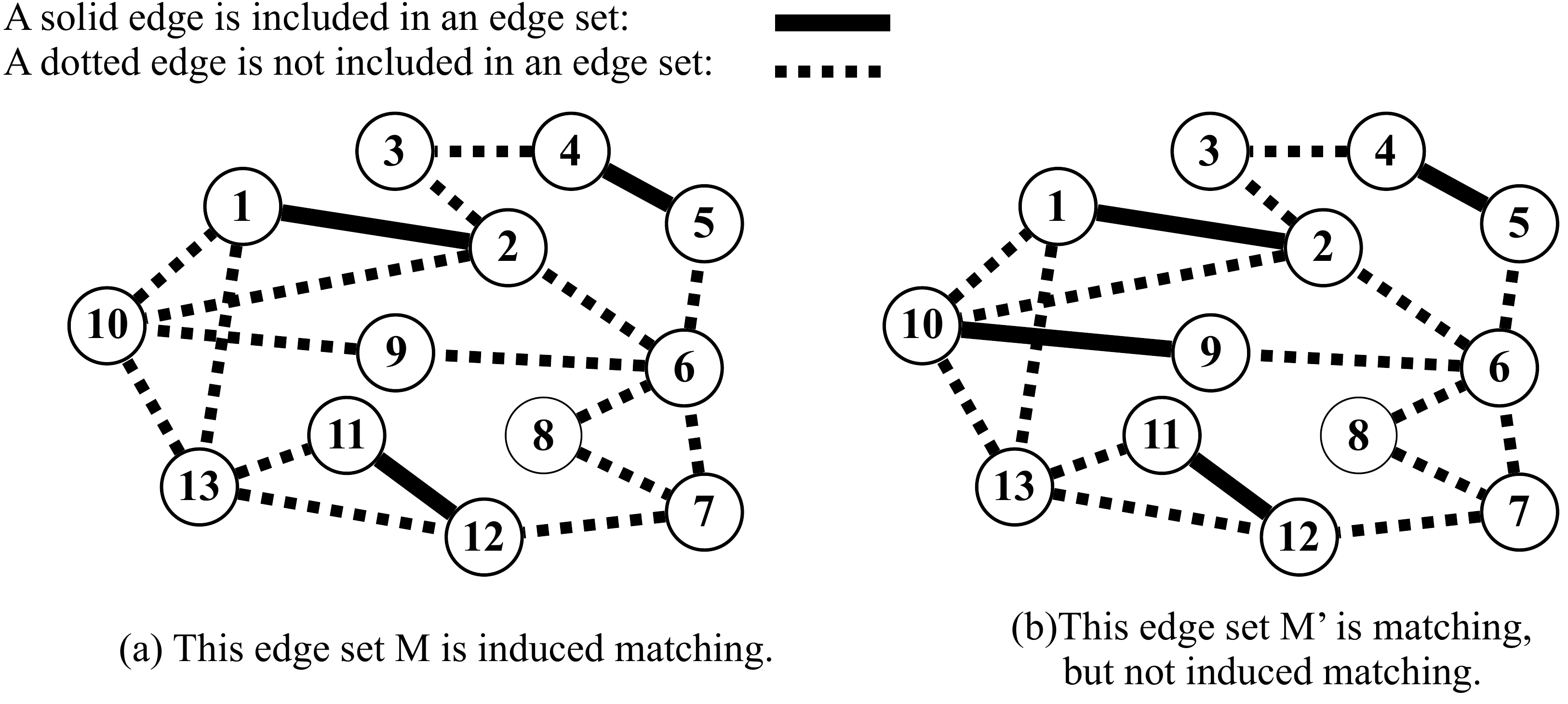}
  \caption{An edge set $M$ in (a) is an induced matching, because all distinct
    two edges $e, f \in M$ hold $dist_G(e, f) \ge 2$.
    An edge set $M'$ in (b) is not an induced matching,
    because edge $e = \set{9, 10}$ and $f = \set{1, 2}$ do not hold
    $dist_G(e, f) \ge 2$. }
  \label{fig:im}
\end{figure}

Now, 
we define the induced matching enumeration problem as follows; 
\begin{problem}[The induced matching enumeration problem]
  Enumerate all induced matching in a given graph $G$ without duplicates. 
\end{problem}

%% file: algoc4free.tex
\section{Enumeration of Induced Matchings for $C_4$-free Graphs}
\label{sec:c4free}

\subsection{Binary partition}
\label{sec:BinaryPartition}

\begin{figure}[t]
    \centering
    \includegraphics[width=70mm]{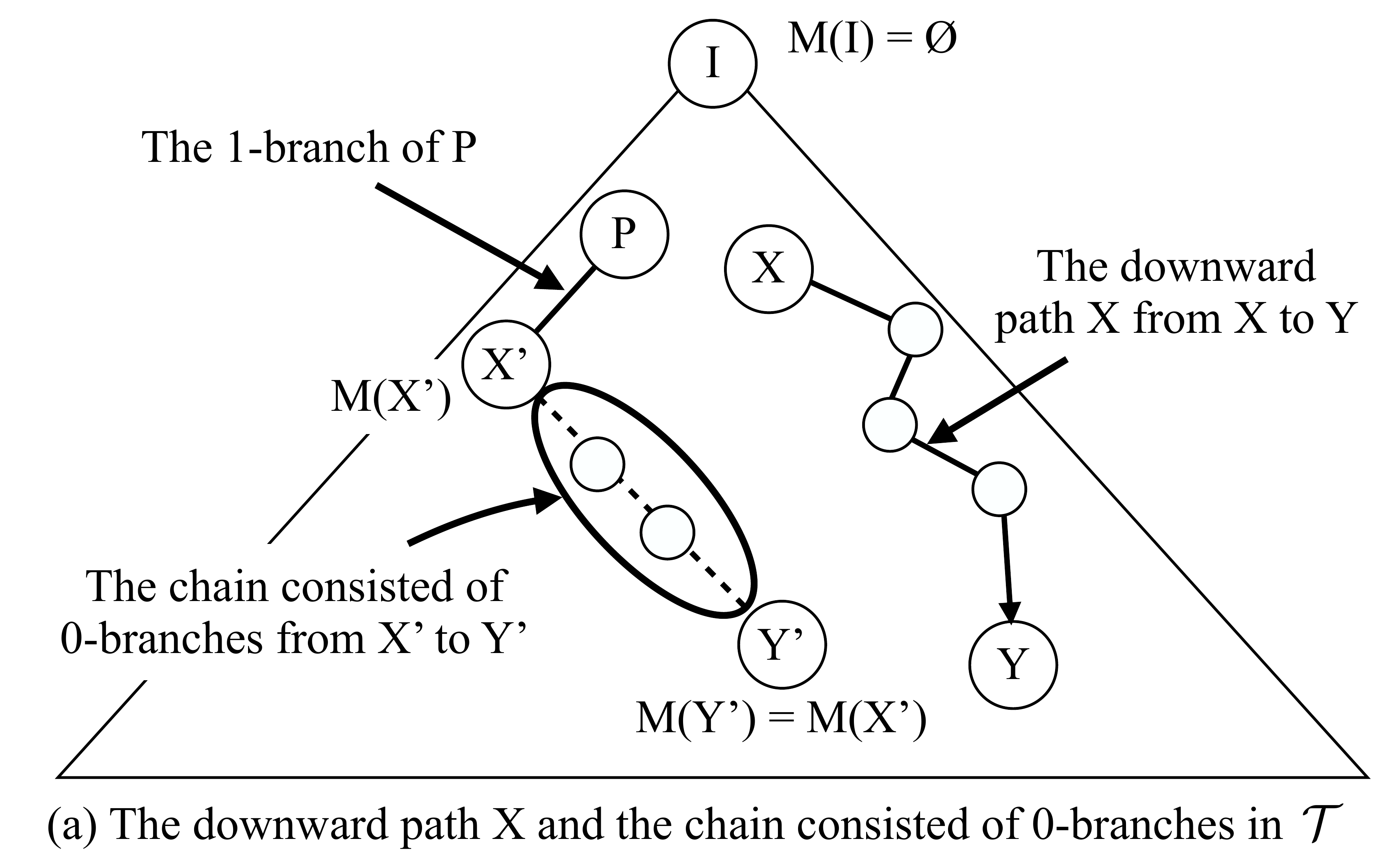}
    \caption{The solid line from $X$ to $Y$ represents
      a downward path.
      Consequently, $X \preceq Y$ holds.
      The dotted line represents a chain consisting of $0$-branches.
      $X'$ is a top of the chain since $X'$ is a $1$-child of $P$.
      $Y'$ is a bottom of the chain since $Y'$ does not have a child. 
    }
    \label{fig:chain}
\end{figure}

A \name{binary partition method} is an algorithm
which enumerates all solutions
by  dividing a search space into two disjoint search spaces recursively.
We call a dividing step an \name{iteration}. 
Let $G$, $\sig M(G)$, and $I$ be 
an input graph, the set of solutions for $G$, and an iteration of the algorithm, respectively. 
Let $\sig S(I)$ be the set of solutions included in a search space of $I$. 
In the initial state, $\sig S(I) = \sig M(G)$ holds. 
At the initial iteration, the algorithm selects an edge $e$ in $G$ such that 
$e$ satisfies $\size{\sig S_0(I)} \ge 1$ and 
$\size{\sig S_1(I)} \ge 1$ where $S_0(I) = \inset{M \in S(I)}{e \notin M}$ and
$S_1(I) = \inset{M \in S(I)}{e \in M}$. 
Note that $\sig S(I) = \sig S_0(I) \sqcup \sig S_1(I)$ holds.
$\sig A$ recursively applies this procedure until all edges are selected. 

Next, 
we introduce a \name{binary enumeration tree} $\sig T(\sig I) = \sig T = (\sig V, \sig E)$, 
for an input $\sig I$.
Here, $\sig V$ is the set of iterations of $\sig A$ for $\sig I$ and  
$\sig E$ is a subset of $\sig V \times \sig V$. 
For any iteration $X$, 
we define the edge set $E_X$ as follows: 
$E_X = \bigcap_{M \in \sig S(X)} M$. 
For any iterations $X$ and $Y$,
$Y$ is a \name{child} of $X$
if $E_Y \subset E_X$ and $\size{E_X \setminus E_Y} = 1$ hold. 
We call $X$ is the \name{parent} of $Y$. 
For any iteration $X$, 
we define iterations $X.1$ and $X.0$ with
the set of solutions $\sig S_1(X)$ and $\sig S_0(X)$, respectively.
That is, $X.1$ and $X.0$ are the children of $X$. 
In particular, 
$X.1$ is the \name{$1$-child of $X$}, and $X.0$ is the \name{$0$-child of $X$}.
In addition, 
we call edges $e = \set{X, X.0}$ and $f = \set{X, X.1}$ in $\sig E$ 
a \name{$0$-branch} and a \name{$1$-branch}, respectively. 
In the binary enumeration tree $\sig T$ for $\sig I$,
we call  an iteration with children an \name{internal iteration}, and
an iteration without children a \name{leaf iteration}.
Moreover, an iteration $X$ is the \name{root iteration} if
there is no iteration that has $X$ as a child, that is, 
$X$ is the first iteration called by $\sig A$.
For the simplicity, 
we call the binary enumeration tree the \name{enumeration tree}. 

For any iterations $X$ and $Y$, 
a \name{downward path} (or an \name{upward path})
from $X$ to $Y$ in $\sig T$ is 
a sequence of iterations
$\sig L = (X = X_0, \dots, X_k = Y)$, 
where for each $i = 1, \dots, k$, 
$X_k$ is a child (or the parent) of $X_{k - 1}$. 
The length of $\sig L$ is defined as $k-1$. 
For any iterations $X$ and $Y$,
if there is a downward path $\sig L$ from $X$ to $Y$, 
then $X$ is an \name{ancestor} of $Y$ and $Y$ is a \name{descendant} of $X$.
$X \preceq Y$ if $X$ is an ancestor of $Y$. 
For any iteration $Y$,
the set $\inset{X}{X \preceq Y}$ is a \name{chain}\cite{Diestel:10} of $Y$. 
When iterations $X$ and $Y$ belong to a same chain,
$X$ and $Y$ are \name{comparable}.
In a chain $\sig L$, 
we call an iteration $X$ is the \name{minimum element} in $\sig L$ and the \name{maximum element} in $\sig L$ 
if $X$ is the head of $\sig L$ and  is the tail of $\sig L$, respectively. 
In Figure~\ref{fig:chain},  
We show an example of a downward path and a chain in the enumeration tree $\sig T$.

\subsection{Algorithm for $C_4$-free graphs}

\begin{algorithm}[t]
  \caption{The algorithm enumerating all induced
    matchings in $C_4$-free graphs in constant amortized time.}
  \label{algo:cat:detail}
  \Procedure{\EnumIM($G = (V, E)$)}{
    \RecEIM$(\emptyset, G)$ \;
  }
  \Procedure{\RecEIM$(M, G)$}{
    \If{$E(G) = \emptyset$}{
      Output $M$\;
      \Return \;
    }
    The vertex $v$ has the maximum degree in $G$\;\label{step:cat:find}
    \RecEIM$(M, G\setminus \set{v})$\tcc*{$0$-child  \label{step:cat:Child0}}
    $G' \gets G\setminus N[v]$\;                  \label{step:cat:R1}
    \For{$e \in D_v(0)$}{                            \label{step:cat:for:select:edge}
      \RecEIM$(M \cup \set{e}, G' \setminus Sect_e(2))$\tcc*{$i$-child}  \label{step:cat:Child1}
    }
    Restore edges in $D_v(0) \cup D_v(1)$\; \label{step:cat:S2}
    \Return \;
  }
\end{algorithm}

In what follows, 
suppose that an input graph is a $C_4$-free graph. 
We show the algorithm \EnumIM in \name{Algorithm}~\ref{algo:cat:detail}. 
For any iteration $X$ of \RecEIM, 
let $M(X)$, $G(X)$ and $S(X)$ be 
the current induced matching as solution, a graph, and 
the set of vertices that are selected in the ancestor iterations of $X$ as the pivot used to partition the problem, respectively.
\RecEIM outputs $M(X)$ as a solution if 
no edge can be added to $M(X)$ from $G(X)$ and quits $X$. 
\RecEIM skips this step and execute the following steps 
if there is an edge that can be added to $M(X)$. 

An edge $e$ in $G$ is a \name{safe edge} if $e$ satisfies the following condition: 
For any edge $f \in M(X)$, $dist(e, f) \ge 2$. 
$e$ is a \name{conflict edge} otherwise.
Let $G(X) = G[V\setminus (N(V(M(X))) \cup S(X))]$. 
That is, $G(X)$ is a graph removed all conflicting edges with $M(X)$ and $S(X)$ from $G$.
\RecEIM firstly selects the vertex $v$ with the maximum degree. 
We call such a vertex $v$ a \name{pivot} on $X$. 
Next, 
\RecEIM divides a solution set $\sig S$ into
$d(v) + 1$ disjoint sets $\sig S_0, \dots, \sig S_{\size{d(v)}}$.
Let $e_i$ be the $i$th edge incident to $v$ for $i \in \set{1, \dots, d_{G(X)}(v)}$. 
$\sig S_i \subseteq \sig S$ is the set of solutions including $e_i$,
and $\sig S_0 = \sig S \setminus \bigcup_{i = 1, \dots, d_{G(X)}(v)} \sig S_i$ is 
the set of solutions not including edges adjacent to $v$. 
$X.i$ denotes 
the $i$th child iteration of $X$ that receives $M(M(X) \cup \set{e_i})$. 
Also, we call $X.0$ \name{type-$0$ child} and
$X.i$ \name{type-$1$ child} for $i \neq 0$. 
We call the branch from $X$ to the type-$0$ child the \name{$0$-branch} and
a branch from $X$ to type-$1$ child a \name{$1$-branch}.
Let $\sig T$ be an enumeration tree made by \EnumIM. 
Note that $\sig T$ is not a binary tree but a multi-way tree. 
In Figure~\ref{fig:enumtree}, we show an example of $\sig T$. 
A proof of the next lemma is shown in Appendix~\ref{app}. 
\begin{lemma}
  \label{lem:del2:mono}
  Let $X$ and $Y$ be any iterations.
  If $X \preceq Y$ 
  then $M(X) \subseteq M(Y)$ and $E(G(X)) \supseteq E(G(Y))$ hold. 
\end{lemma}

\begin{figure}[t]
  \centering
  \includegraphics[width=0.8\textwidth]{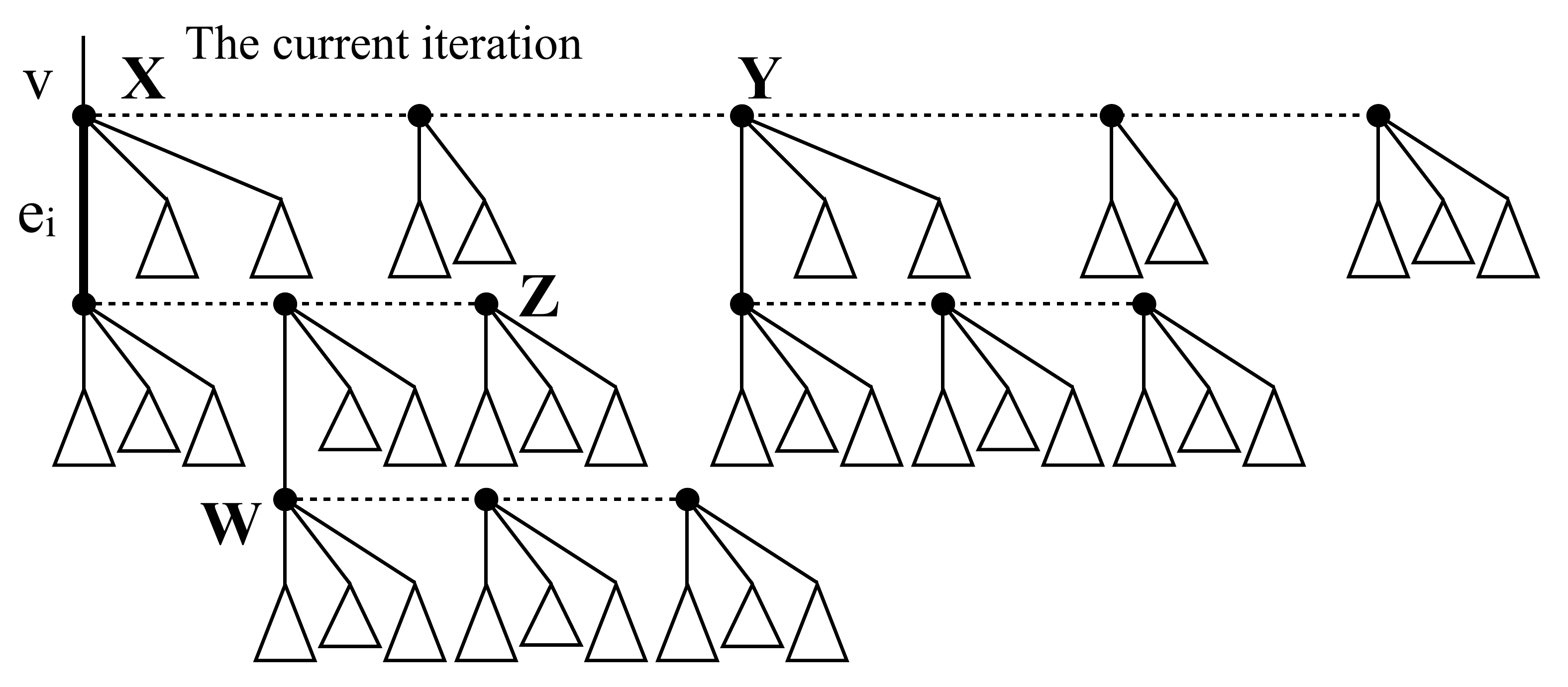}
  \caption{}
  \label{fig:enumtree}
\end{figure}

\subsection{Correctness of the algorithm}
For any iteration $X$ on $\sig T$,
it is redundant to process safe and conflict edges independently among siblings of the same parent. 
To avoid this,
we simultaneously process these edges by
using following concentric structures around the pivot on $X$. 
Let $k \in \set{0, 1, 2}$ and $\ell \in \set{0, 1, 2}$. 
We define the concentric structure $D_v(k, \ell)$ as follows: 
\begin{equation}
  D_v(k, \ell) = \inset{\set{x, y} \in E(G(X))}{dist_{G(X)}(x, v) = k, dist_{G(X)}(y, v) = \ell}. 
\end{equation}
$D_v(k)$ denotes $D_v(k, k) \sqcup D_v(k, k + 1)$. 
We call an edge $e \in D_v(k, \ell)$ a \name{$k$-$\ell$ edge of $v$}  
and  an edge $e \in D_v(k)$ a \name{$k$-$*$ edge of $v$}. 
Figure~\ref{fig:partition} (a) shows
an example of $0$-$1$ edges, $1$-$*$ edges, and $2$-$*$ edges.
The distance between two vertices is defined by the length of
shortest path in not $G$ but $G(X)$. 
The next lemma implies that $M(X) \cup \set{e}$ is an induced matching in $G$ 
for any edge $e \in D_v(0)$. 
A proof of the next lemma is shown in Appendix~\ref{app}. 
\begin{lemma}
  \label{lem:c4free:dist}
  Let $G$ be an input graph, 
  $X$ be any iteration in the algorithm \EnumIM in Algorithm~\ref{algo:cat:detail}, and $e$ be any edge in $D_v(0)$. 
  Then, $M(X) \cup \set{e}$ is an induced matching in $G$. 
\end{lemma}


Since \EnumIM outputs a solution in a leaf iteration and
$M(X)$ is induced matching for each iteration $X \in \sig T$,
the following corollary holds from
Lemma~\ref{lem:c4free:dist}. 

\begin{corollary}
  \label{cor:induced}
  The algorithm \EnumIM in Algorithm~\ref{algo:cat:detail} outputs only induced matchings. 
\end{corollary}

Next, we consider a method for obtaining $G(X.i)$. 
For any $0$-$1$ edge $e_i = \set{v, u_i}$ of a pivot $v$,
we define $Sect_{e_i}(k)$ as follows:
\begin{equation}
    Sect_{e_i}(k) = \inset{f \in D_v(k)}{dist_{G(X)}(v, f) = k, dist_{G(X)}(e_i, f) = k - 1}. 
\end{equation}
We show an example of $Sect_{e_i}(k)$ in Figure~\ref{fig:partition}.
Proofs of the following two lemmas are shown in Appendix~\ref{app}. 
\begin{lemma}
    \label{lem:G:X.0}
    Let $X$ be an iteration in the algorithm \EnumIM in Algorithm~\ref{algo:cat:detail}.
    Then, $G(X.0) = G(X) \setminus \set{v}$ holds. 
\end{lemma}

\begin{lemma}
  \label{lem:G:X.i}
  Let $X$ be any iteration in the algorithm \EnumIM in Algorithm~\ref{algo:cat:detail} and $i$ be positive integer.
  Then, $G(X.i) = G(X) \setminus (D_v(0) \cup D_v(1) \cup Sect_{e_i}(2))$ holds. 
\end{lemma}

Lemma~\ref{lem:G:X.0} and Lemma~\ref{lem:G:X.i} imply 
that \EnumIM correctly compute $G(X.i)$ in $X$.

\begin{lemma}
  \label{lem:no_dup}
  The algorithm \EnumIM in Algorithm~\ref{algo:cat:detail} outputs solutions without duplication. 
\end{lemma}

\begin{proof}
  Let $X$ and $Y$ be two distinct leaf iterations.
  We proceed by contradiction.
  Suppose that $M(X) = M(Y)$.
  From the assumption,  $X$ and $Y$ are incomparable. 
  Hence, without loss of generality the lowest common ancestor of $X$ and $Y$ always exists.
  Let $Z$ be the lowest common ancestor of $X$ and $Y$.
  We consider the following two cases. 
  (1) Suppose that both $X$ and $Y$ are descendants of the type-$0$ child $Z.0$ of $Z$. 
  This contradicts that $Z$ is the lowest common ancestor of $X$ and $Y$.  
  (2) Suppose that at least one of $X$ and $Y$ is a descendant of a type-$1$ child of $Z$.
  Without loss of generality, 
  $X$ is a descendant of the $i$th child $Z.i$ of $Z$. 
  If $W$ is $Z.j$ or is any descendant of $Z.j$ where $j \neq i$, 
  then $M(W)$ does not include $e_i$ and this contradicts $M(X) = M(Y)$. 
  Hence, 
  the statement holds. 
\end{proof}

\begin{lemma}
  \label{lem:output_all}
  The algorithm \EnumIM in Algorithm~\ref{algo:cat:detail} outputs all solutions in $\sig M$. 
\end{lemma}

\begin{proof}
  Let $\sig T$ be an enumeration tree, $X$ be any iteration on $\sig T$,
  and $v$ be the pivot on $X$. 
  From Lemma~\ref{lem:G:X.0} and Lemma~\ref{lem:G:X.i},
  \EnumIM correctly divides a solution set $\sig S(X)$ in $X$ into
  $\sig S_0(X), \dots, \sig S_{d(v)}(X)$. 
  If $\size{\sig S(X)} = 1$,
  then \EnumIM outputs $\sig S(X)$. 
  Hence, \EnumIM outputs all solutions in $\sig M$. 
\end{proof}

From corollary~\ref{cor:induced}, Lemma~\ref{lem:no_dup},
and, Lemma~\ref{lem:output_all}, the next theorem holds. 
\begin{theorem}
  \label{theo:c4correct}
  The algorithm \EnumIM in Algorithm~\ref{algo:cat:detail} enumerates all solutions without duplication. 
\end{theorem}

\begin{figure}[t]
  \centering
  \includegraphics[width=\textwidth]{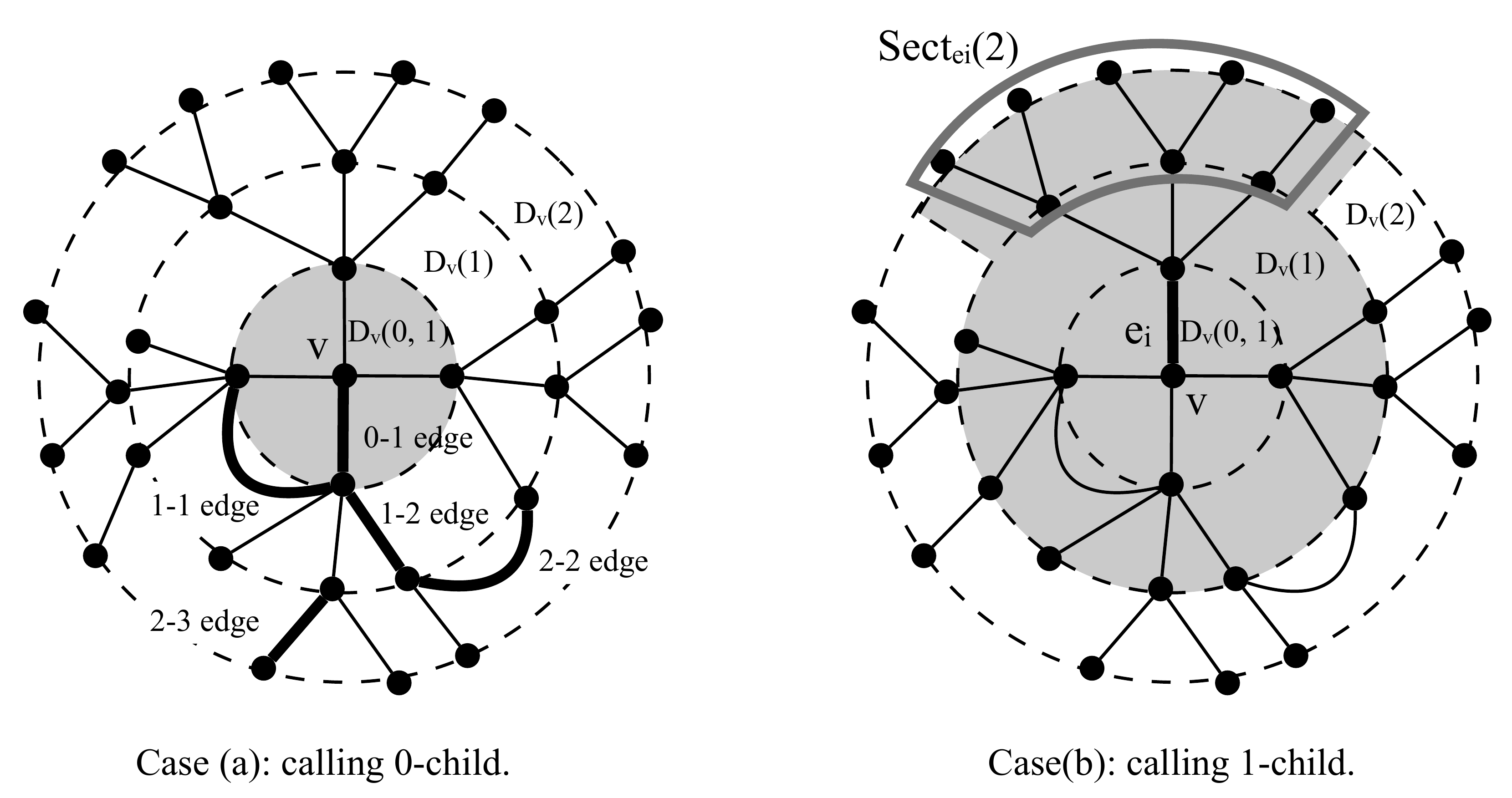}
  \caption{An example of
    dividing edges by the pivot $v$ in the graph $G$.
    In case (a), the shaded area indicates the area of edges to be
    removed when calling type-$0$ child. 
    In case (b), the shaded area indicates the area of edges to be
    removed when calling type-$1$ child.
    The area surrounded by the solid line represent
    $Sect_{e_i}(2)$}
  \label{fig:partition}
\end{figure}

\subsection{Amortized analysis of the time complexity}
\label{sec:c4free:cost}

If the degree of the pivot on $X$ is less than three, 
then
\EnumIM obviously runs in  the constant amortized time per solution 
since the number of steps in each iteration $X$ is constant.
Thus, we assume the degree of pivot is at least three.

We first consider the data structure $\List{G(X)}$ 
to efficiently extract the set of vertices whose degree is $k$ when we are given $k$. 
We define $\List{G(X)}$ as follows:
$\List{G(X)} = \set{L_0, \dots, L_{\Delta(G(X))}}$, 
where  for any $0 \le i \le \Delta(G(X))$,
$L_i = \inset{v \in V(G)}{d_{G(X)}(v) = i}$. 
The lists in $\List{G(X)}$ are implemented by doubly-linked lists,
and  we denote $x \in \List{G(X)}$ if $\List{G(X)}$ includes $x$.
For any input graph $G$, 
we can compute  $\List{G}$ in $\order{\size{V(G)}}$ time
by using bucket sort.
We can implement $\List{G(X)}$ such that 
extracting any vertex from $\List{G(X)}$ can be done in constant time. 
By using $\List{G(X)}$, 
we can see the following lemma. 
The proof of Lemma~\ref{lem:Lconst} in Appendix~\ref{app}. 
\begin{lemma}
\label{lem:Lconst}
  Let $X$ be any iteration in $\sig T$.
  Then, we can find the pivot in constant time by using $\List{G(X)}$. 
\end{lemma}

Next, we define the edge set $D^\le_v(2)$ as follows:
$D^\le_v(2) = D_v(0) \cup D_v(1) \cup D_v(2)$, i.e.
$D^\le_v(2)$ is the edge set consisting of all edges whose
distance is less than two from $v$.

\begin{lemma}
  \label{lem:c4}
  Let $G$ be a $C_4$-free graph, $v$ be the pivot on an iteration $X$, and,
  $u$ be a vertex satisfying $dist_G(u, v) = 2$. 
  Then, the number of edges whose endpoint is $u$ in the set of $1$-$2$ edges of $v$ is
  exactly one.  
\end{lemma}

\begin{proof}
  We proof by contradiction.
  Let $f_1 = \set{u, w_1}$ and $f_2 = \set{u, w_2}$ be two distinct $1$-$2$ edges whose end point is $u$.
  We assume $w_1 \neq w_2$.
  Since $f_1$ and $f_2$ are $1$-$2$ edges and $dist_G(u, v) = 2$, 
  $dist_G(v, w_1) = dist_G(v, w_2) = 1$ holds. 
  Thus, there exist two edges $e_1 = \set{v, w_1}$ and $e_2 = \set{v, w_2}$. 
  Hence, there is a cycle $(v, e_1, w_1, f_1, u, f_2, w_2, e_2, v)$ in $G$. 
  This contradicts that $G$ is a $C_4$-free graph. 
  Hence, the statement holds. 
\end{proof}

\begin{lemma}
  \label{lem:cat}
  Let $G$ be a $C_4$-free graph and
  $v$ be the pivot on an iteration. 
  Then, the following inequality holds: 
  $\sum_{e \in D_v(0)} \size{Sect_e(2)} \le 2\size{D_v(2)}$. 
\end{lemma}

\begin{proof}
  We show that $\bigcup_{e \in D_v(0)} Sect_e(2) = D_v(2)$.
  Let $f = \set{x, y}$ be an edge in $Sect_e(2)$.
  By definition, $f \in D_v(2)$ holds.
  Without loss of generality, 
  we can assume that $dist_G(x, v) = 2$. 
  Since $dist_G(x, v) = 2$, 
  there is a vertex $w$ satisfying $dist_G(x, w) = dist_G(w, v) = 1$. 
  By definition, $f$ belongs to $Sect_{\set{w, v}}(2)$.
  Hence, $\bigcup_{e \in D_v(0)} Sect_e(2) = D_v(2)$ holds. 

  Next, we assume that for any $2$-$*$ edge $f \in D_v(2)$, 
  $f$ belongs to the following three sets; 
  $Sect_{e_1}(2)$, $Sect_{e_2}(2)$, and $Sect_{e_3}(2)$
  , where $e_1, e_2, e_3 \in D_v(0)$. 
  Then, by the definition of $Sect_{e_i}(2)$, 
  $dist_G(e_i, f) = 1$ holds for $i \in \set{1, 2, 3}$.
  Hence, there is some $g_i = \set{x_i, y_i}$ 
  satisfying  $x_i \in e_i$ and $y_i \in f$. 
  By the definition of $e_i$,
  $dist_G(v, x_i) = 1$ and $dist_G(v, y_i) = 2$ hold.
  This implies that $g_i$ is a $1$-$2$ edge
  that shares the end point with $f$. 
  By the pigeonhole principle, 
  one of the end points of $f$ has at least  two $1$-$2$ edges.
  This contradicts with Lemma~\ref{lem:c4}, 
  hence the statement holds. 
\end{proof}

In the remaining of this section, 
we show that \EnumIM enumerates all solutions in constant amortized time per solution. 
To show the complexity, 
we show that the ratio between the number of $1$-child iterations and $0$-child iterations is constant.  
If the statement holds, 
then the number of iterations on $\sig T$ is linear in the number of leaf iterations of $\sig T$.
Let $X$ be any iteration in $\sig T$ and $v$ be the pivot on $X$.
Suppose that $e = \set{x, y}$ is any edge in $\in D^\le_v(2)$ and $f = \set{v, x}$. 
We denote by $C(X, e)$ a descendant iteration  of $X$
such that $Y = C(X, e)$ is the top of the chain including $X$ and receives $M(Y)$ defined as follows. 

\renewcommand{\labelenumi}{(\arabic{enumi})}
\begin{enumerate}
    \item[(1)] If $e$ is a $0$-$1$ edge, then $M(Y) = M(X) \cup \set{e}$.
    \item[(2)] If $e$ is a $1$-$1$ edge, then $M(Y) = M(X) \cup \set{f}$.
    \item[(3.a)] If $e$ is a $1$-$2$ edge and $Sect_f(2) = \emptyset$,
               then $M(Y) = M(X) \cup \set{e}$.
    \item[(3.b)] If $e$ is a $1$-$2$ edge and $Sect_f(2) \neq \emptyset$, 
               then $M(Y) = M(X) \cup \set{f', g}$, where  
               $f' \in D_v(0)$ and  $g \in Sect_f(2)$ such that 
               $f' \neq f$ and $dist_G(f', g) = 2$. 
    \item[(4)] If $e$ is $2$-$*$ edge,
               then $M(Y) = M(X) \cup \set{e, h}$, where $h$ is an edge 
               such that $h$ is adjacent to $v$ and $dist_G(e, h) = 2$. 
\end{enumerate}
We call $C(X, e)$ the \name{corresponding iteration to $X$ w.r.t $e$}. 
The next lemma shows that $C(X, e)$ satisfying the above conditions always exists. 

\begin{lemma}
  \label{lem:charge}
  For any iteration $X$ and $e \in D^\le_v(2)$, 
  there always exists the corresponding iteration $C(X, e)$ to $X$ w.r.t $e$. 
\end{lemma}

\begin{proof}
  Let $e = (x, y)$.
  From Lemma~\ref{lem:output_all},
  to proof the lemma,  all we have to do is show that $M(C(X, e))$ is a solution. 
  If (1) or (3.a) holds, then $M(C(X, e))$ is obviously an induced matching since $e \in D^\le_v(2)$.  
  Next, we consider condition (2). 
  There are two edges $\set{v, x} = f$ and $\set{y, v}$ since $e$ is a $1$-$1$ edge.
  Since $f \in D^\le_v(2)$, $M(C(X, e))$ is a solution. 
  Next, we consider condition (3.b).
  Let $f = \set{v, x}$.
  Since $e$ is a $1$-$2$ edge, such edge $f$ always exists.
  Let $g$ be any edge in $Sect_f(2)$.
  From Lemma~\ref{lem:c4},
  at least one of the end points of $g$ connects exactly one $1$-$2$ edge.
  Hence, there is an edge $f' \neq f$ that is adjacent to $v$ and satisfies
  $dist_G(f', g) = 2$ since the degree of $v$ is at least three.
  Moreover,  $\set{f', g}$ is an induced matching 
  since $dist_G(f', g) = 2$. 
  Hence, $M(C(X, e))$ is an induced matching. 
  Finally, we consider condition (4).
  Since $d_G(v) \ge 3$, there exists an edge $h$ that is adjacent to $v$ and satisfies $dist_G(e, h)$.
  Hence, $M(C(X, e))$ is an induced matching. 
\end{proof}

In the following lemmas, 
for any iteration $X$, 
we show the number of pairs of an iteration and an edge 
whose corresponding iteration is $X$ is constant.

\begin{lemma}
  \label{lem:one}
  If a graph $G$ is $C_4$-free, 
  then the number of $1$-$1$ edges adjacent to $0$-$1$ edges is at most one. 
\end{lemma}

\begin{proof}
  We show the lemma by contradiction. 
  Suppose there are two distinct $1$-$1$ edges 
  $f = \set{u, w}$ and $g = \set{u, x}$ that are  adjacent to a $0$-$1$ edge $e$. 
  By the definition of a $1$-$1$ edge, 
  $\set{u, w, x} \subseteq N(v)$.
  Hence, there exist two distinct edges $f' = \set{v, w}$ and $g' = \set{v, x}$. 
  However, this implies that there exist a $4$-cycle 
  $(v, f', w, f, u, g, x, g', v)$.
  This contradicts that $G$ is $C_4$-free. 
  Hence, the statement holds. 
\end{proof}

For the proofs of the next lemmas, see Appendix~\ref{app}. 

\begin{lemma}
  \label{lem:c3a}
  Let $X$ and $e$ be 
  a pair of an iteration on $\sig T$ and an edge in $D_v^\le(2)$ 
  satisfying condition (3.a), 
  and $Y$ be any iteration on $\sig L$ from $X$ to $C(X, e)$
  such that $Y$ satisfies $C(Y, e) = C(X, e)$.  
  Then, the number of such $Y$  is at most two. 
\end{lemma}

\begin{lemma}
  \label{lem:corresponded_rec}
  Let $X$ be any iteration in $\sig T$.
  Then, the number of pairs an iteration $Y$ and an edge $e$ satisfying 
  $C(Y, e) = X$ is at most six. 
\end{lemma}


From Lemma~\ref{lem:one}, Lemma~\ref{lem:c3a}, and
Lemma~\ref{lem:corresponded_rec}, 
for any iteration $X \in \sig T$, 
the number of pairs of an iteration $Y$ and an edge $e$ such that $C(Y, e) = X$ 
is constant.
Next, the following lemmas show that total computation time in \EnumIM is
$\order{\size{\sig T}}$ time. 
Let $F(X)$ be $\bigcap_{i = 0, \dots, \Delta(G(X))} E(G(X.i))$.
That is, $F(X)$ is the set of edges that are shared by all child iterations of $X$. 

\begin{lemma}
  \label{lem:fx}
  Let $v$ be the pivot on an iteration $X$ in $\sig T$. 
  Then, $E(G(X)) \setminus F(X) = D^\le_v(2)$. 
\end{lemma}

\begin{proof}
  Let $\set{e_1, \dots, e_{d_{G(X)}(v)}}$  be 
  the set of edges that are adjacent to $v$. 
  We show $F(X) = E(G(X)) \setminus D^\le_v(2)$.
  Firstly, we show $F(X) \subseteq E(G(X)) \setminus D^\le_v(2)$. 
  For any $i = 0, \dots, d_{G(X)}(v)$,
  $M(G(X.i))$ includes $e_i = \set{v, u_i}$ by definition.  
  Hence, $F(X)$ does not include conflicting edges of $e_i$.
  In addition, each edge $f \in F(X)$ satisfies
  $dist_{G(X)}(f, v) \ge 2$ and $dist_{G(X)}(f, u_i) \ge 2$. 
  Therefore, $f$ is not included in $D^\le_v(2)$. 
  Secondly, we show $F(X) \supseteq E(G(X)) \setminus D^\le_v(2)$. 
  Let $g$ be any edge in $E(G(X)) \setminus D^\le_v(2)$.
  By definition, $dist_{G(X)}(g, e_i) \ge 2$ holds for any $e_i$. 
  Therefore, $g \in F(X)$ since $g \in E(G(X.i))$. 
  Now, $D^\le_v(2) \subseteq E(G(X))$ and  $F(X) \cap D^\le_v(2) = \emptyset$. 
  Hence, the statement holds. 
\end{proof}

\begin{lemma}
  \label{lem:tree_size}
  $\sum_{X \in V(\sig T)} \size{E(G(X)) \setminus F(X)}$ is bounded by 
  $\order{\size{\sig T}}$.
\end{lemma}

\begin{proof}
    Let $X$ be any iteration and $v$ be the pivot on $X$. 
    The number of iterations $Y = C(X, e)$ is at most $\size{D^\le_v(2)}$, 
    where $e$ is an edge in $D^\le_v(2)$.
    Hence, the number of all corresponding iterations is equal to $\sum_{X \in V(\sig T)} \size{E(G(X)) \setminus F(X)}$
    since $E(G(X)) \setminus F(X) = D^\le_v(2)$ from Lemma~\ref{lem:fx} 
    together with that  a pair of an internal iteration $X$ and an edge $e \in D^\le_v(2)$ corresponds to
    exactly one iteration $C(X, e)$. 
    Next, 
    we consider the number of all corresponding iterations. 
    The number of pairs of an iteration $Y$ and an edge $e$ such that $C(Y, e) = X$ is at most constant from Lemma~\ref{lem:corresponded_rec}. 
    Since the number of internal iteration is $\size{\sig T}$, 
    the number of all corresponding iterations is bounded by $\order{\size{\sig T}}$, 
    Hence, the statement holds. 
\end{proof}

\begin{theorem}
  The algorithm \EnumIM in Algorithm~\ref{algo:cat:detail} enumerates all induced matchings
  in constant amortized time per solution in a $C_4$-free graph $G$ 
  after $\order{\size{V} + \size{E}}$ preprocessing time. 
\end{theorem}

\begin{proof}
    The correctness of \EnumIM is obvious from Lemma~\ref{theo:c4correct}.
    Next, we consider the time complexity of \EnumIM. 
    Let $X$ be an iteration of \EnumIM. 
    In the preprocessing phase,  
    \EnumIM constructs $\List{G}$ in $\order{\size{V} + \size{E}}$ time by using bucket sort.
    Next, we consider the total time for deleting edges in an input graph. 
    Each edge is deleted at most twice from Lemma~\ref{lem:c4} in each iteration. 
    Moreover, 
    from Lemma~\ref{lem:tree_size},
    the total number of deleted edges is $\order{\size{\sig T}}$ in \EnumIM. 
    Hence, the total time of edge deletion is $\order{\size{\sig T}}$ time 
    since each edge can be removed in constant time from the input graph. 
    Next, we consider the total time of the updating $\List{G(X)}$. 
    When \EnumIM removes an edge $e = \set{u, v}$ from $X$,
    \EnumIM moves $u \in L_i$ to $L_{i - 1}$ and $v \in L_j$ to $L_{j - 1}$. 
    Since it can be done in constant time,
    the time complexity of \EnumIM is $\order{\size{\sig T}}$. 
    In addition, 
    every iteration $X$ in $\sig T$ has a child at least two.
    Hence, the number of solutions is $\Omega(\size{\sig T})$. 
    Therefore, \EnumIM runs in $\order{\size{\sig T} / \size{\sig T}} = \order{1}$ time per solution. 
\end{proof}

%% file: algomso.tex
\section{Counting of Induced Matchings for Other Graph Classes}
\label{secmso}

\newtheorem{proposition}[lemma]{\textbf{Proposition}} 
\newcommand{\Incd}{\textit{IND}}


To complement the result of Sec.~\ref{sec:c4free}, in this section, we present some fixed-parameter tractability results on
counting the number of induced matchings
in terms of descriptive complexity theory%
~\cite{flum2006parameterized}.
See Appendix~\ref{app:fo} for omitted definitions and proofs.  
Recently, Frick~\cite{frick:TCSyst:2004generalized} introduced the notion of \name{locally tree-decomposability} by generalizing the tree decomposition.
He showed the fact that the classes of graphs of bounded degree, of bounded tree-width, and planar graphs are locally tree-decomposable~\cite{frick:TCSyst:2004generalized}. 


\begin{proposition}[Frick~\cite{frick:TCSyst:2004generalized}]
  \label{prop:fpt:fo:frick}
  Let $\sig C$ be any class of locally tree-decomposable structures.
  For any structure $\sig{A} \in \sig C$, a counting problem $\Pi$ definable in FO can be solved in linear time in $||\sig{A}||$,
  where $\sig{A}$ is given with its underlying
  nice tree cover
  $\sig T$ associated with $r, \ell, g$. 
\end{proposition}


For any $k\ge 0$, a \name{$k$-induced matching} is an induced matching $M \subseteq E$ with $|M| = k$.
From the above fact 
and Proposition~\ref{prop:fpt:fo:frick}, we show the next theorem. 

\begin{theorem}
  \label{thm:im:fo}
  For any class $\sig G$ of graphs of bounded degree, graphs of bounded tree-width, or planar graphs and any $k\ge 0$, the counting problem of $k$-induced matchings in an input graph $G$ in $\sig G$ can be solved in linear time in $||G||$.
\end{theorem}

In the proof, we built a FO-formula $\phi$ describing that an edge subset is a $k$-induced matching. Hence, the counting problem of $k$-induced matchings belongs to FPT when parameterized with some constants determined by $k$, $\phi$, and $\sig G$.

%% file: conc.tex
\section{Conclusion}
\label{sec:conc}
In this paper, we presented an efficient algorithm for enumerating
all induced matchings in constant amortized time for $C_4$-free graphs.
Generalization of this result to other graph classes is an interesting
 future work.
Investigating the other class of induced subgraphs such as induced paths
 is also interesting.
We also have interests in the independent set enumeration in the square of line graphs\cite{Cameron:DAM:1989} and also counting problems of these structures.


%% file: acknowledge.tex
\subsubsection*{Acknowledgements}
This research was supported by
Grant-in-Aid for Scientific Research(A)
Number 16H01743.

%% file: appendix.tex
\section{Appendix (Proof of some lemmas)}
\label{app}

In this appendix, we show proofs of 
Lemma~\ref{lem:del2:mono}, 
Lemma~\ref{lem:c4free:dist}, 
Lemma~\ref{lem:G:X.0}, 
Lemma~\ref{lem:G:X.i}, and
Lemma~\ref{lem:Lconst}. 

\begin{lemma}
  \label{app:del2:mono}
  Let $X$ and $Y$ be any iterations.
  If $X \preceq Y$ 
  then $M(X) \subseteq M(Y)$ and $E(G(X)) \supseteq E(G(Y))$ hold. 
  (A proof of Lemma~\ref{lem:del2:mono}. )
\end{lemma}

\begin{proof}
  There is a downward path $\sig L$ from $X$ to $Y$ since $X \preceq Y$.  
  It is obvious that
  \EnumIM does not remove
  any edge in $M(X)$ from $M(X')$ in $X' \in \sig L$.
  On the other hands, 
  \EnumIM does not add any edge to $E(G(X))$ to $E(G(X'))$ in $X' \in \sig L$.  
\end{proof}

\begin{lemma}
  \label{app:c4free:dist}
  Let $G$ be an input graph, 
  $X$ be any iteration in the alglrithm \EnumIM in Algorithm~\ref{algo:cat:detail}, and $e$ be any edge in $D_v(0)$. 
  Then, $M(X) \cup \set{e}$ is an induced matching in $G$. 
  (A proof of Lemma~\ref{lem:c4free:dist}. )
\end{lemma}

\begin{proof}
  We proof by contradiction. 
  Assume that there is a conflicting edge $f \in M(X)$ with $e$.
  From Lemma~\ref{lem:del2:mono},
  there is an iteration $Y \preceq X$ such that $f$ is added to $M(Y)$.
  Since $f$ conflict with $e$,
  either (1) $e$ is adjacent to $f$ or
  (2) $e$ is not adjacent to $f$ and 
  there is an edge $g$ in $G(Y)$ that is adjacent to both  $e$ and $f$. 
  We consider case (1).
  By the definition of $G(Y)$,
  if \EnumIM adds $f$ to $M(Y)$ in $Y$,
  edges adjacent to $f$ are not in $G(W)$,
  for any descendant iteration $W$ of $Y$. 
  This contradicts the assumption. 
  Next, we consider case (2).
  If $g \in G(Y)$ holds,  then $dist_{G(Y)}(e, g) = 1$.
  This implies that for any descendant iteration $W$ of $Y$,
  $e \notin G(W)$.
  On the other hands, 
  if $g \notin G(Y)$, then $G(Y)$ is not induced subgraph since 
  $e, f \in G(Y)$ and $g \notin G(Y)$.
  This also contradicts the assumption.
  Hence the statement holds. 
\end{proof}

\begin{lemma}
    \label{app:G:X.0}
    Let $X$ be an iteration in the algorithm \EnumIM in Algorithm~\ref{algo:cat:detail}.
    Then, $G(X.0) = G(X) \setminus \set{v}$ holds. 
    (A proof of Lemma~\ref{lem:G:X.0}. )
\end{lemma}
\begin{proof}
    Since $M(X) = M(X.0)$, 
    there is no edge in $G(X.0)$ such that one of its endpoint is $v$. 
    Thus, the statement holds. 
\end{proof}

\begin{lemma}
  \label{app:G:X.i}
  Let $X$ be any iteration in the algorithm \EnumIM in Algorithm~\ref{algo:cat:detail} and $i$ be positive integer.
  Then, $G(X.i) = G(X) \setminus (D_v(0) \cup D_v(1) \cup Sect_{e_i}(2))$ holds. 
  (A proof of Lemma~\ref{lem:G:X.i}. )
\end{lemma}
\begin{proof}
  For any edge $f$ in $G$, 
  we show that the following cases are equivalent: 
  (1) $f$ is conflicting edge of $e$ and 
  (2) $f$ belongs to $D_v(0) \cup D_v(1) \cup Sect_{e_i}(2)$. 
  Let $e_i = \set{v, v_i}$ and $f = \set{u, w}$. 
  Without loss of generality, 
  we can assume that $dist_G(u, v) \le dist_G(w, v)$. 
  If $f \notin D_v(0) \cup D_v(1) \cup Sect_{e_i}(2)$, 
  then $1 < dist_G(u, v) \le dist_G(w, v)$ and
  $1 < dist_G(u, v_1) \le dist_G(w, v_1)$ hold.
  Hence, $f$ does not conflict with $e_i$. 
  On the other hands,
  if $f \notin D_v(0) \cup D_v(1) \cup Sect_{e_i}(2)$ 
  then $f$ obviously conflicts with $e_i$. 
  Hence, 
  the statement holds. 
\end{proof}

\begin{lemma}
\label{app:Lconst}
  Let $X$ be any iteration in $\sig T$.
  Then, we can find the pivot in constant time by using $\List{G(X)}$. 
  (A proof of Lemma~\ref{lem:Lconst}. )
\end{lemma}

\begin{proof}
    It is obvious that  $\bigsqcup_{L_i \in \List{G(X)}} = V(G)$ holds.
    $i_* = \argmax_{i \in \set{0, \dots, \Delta(G(X))}}(L_i \neq \emptyset)$,
    that is, $L_{i_*}$ is the non empty list with the maximum index in $\List{G(X)}$. 
    Now, the pivot on $X$ is in $L_{i_*}$ since $\Delta(G(X)) = i_*$.
    Hence, \EnumIM can obtain $v$ in constant time 
    by extracting a vertex from the tail of $L_{i_*}$. 
\end{proof}


\begin{lemma}
  Let $X$ and $e$ be 
  a pair of an iteration on $\sig T$ and an edge in $D_v^\le(2)$ 
  satisfying condition (3.a), 
  and $Y$ be any iteration on $\sig L$ from $X$ to $C(X, e)$
  such that $Y$ satisfies $C(Y, e) = C(X, e)$.  
  Then, the number of such $Y$  is at most two.
    (A proof of Lemma~\ref{lem:c3a}. )
\end{lemma}
\begin{proof}
  We first show that
  $\sig L$ consists of $0$-branches except the end of $\sig L$. 
  By the definition of $C(X, e)$,
  the end of $\sig L$ is a $1$-branch since $C(X, e)$ is the top of a chain. 
  Next, $\sig L$ includes exactly one $1$-branch 
  since $M(C(X, e)) = M(X) \cup \set{e}$. 
  Hence,  $\sig L$ consists of only $0$-branch other than the end of $\sig L$. 
  
  By using the above observation, 
  we proof by contradiction. 
  Suppose that there exists three distinct iterations $Y_1$, $Y_2$, and $Y_3$
  such that they satisfy condition (3.a) and $C(X, e) = C(Y_1, e) = C(Y_2, e) = C(Y_3, e)$.
  Thus, 
  for any $i = 1, 2, 3$,
  $e$ is a $1$-$2$ edge of $v_i$ that is the pivot of $Y_i$.
  Let $e = \set{x, y}$ and $f_i$ be an edge such that $f_i$ shares the end point
  with $e$ and $f_i$ is adjacent to $v_i$.
  Without loss of generality, 
  we can assume $f_1 = \set{x, v_1}$ and $Y_1 \preceq Y_2 $, $Y_1 \preceq Y_3$. 
  For $j = 2, 3$, 
  if $G(Y_j)$ has the edge $f_j = \set{y, v_j}$, 
  then this contradict with $Sect_{f_i}(2) = \emptyset$.
  Hence, $v_j$ is a neighbor of $x$.
  In $Y_1$, the number of $1$-$1$ edges adjacent to $f_1$ is
  at most one from Lemma~\ref{lem:one}. 
  Hence, at least one of $f_2$ and $f_3$ is a $1$-$2$ edge.
  Without loss of generality, 
  we can assume that $f_2$ is a $1$-$2$ edge. 
  Since $Sect_{f_1}(2) = \emptyset$,
  $D_{G(Y_1)}(v_3) = 1$.
  In addition, 
  in $Y_3$, $d_{G(Y_3)}(x) \ge 2$ 
  since $x$ is adjacent to $y$ and $v_3$. 
  However, this contradicts with $d_{G(Y_3)}(v_3) \ge d_{G(Y_3)}(x)$. 
  Hence, the statement holds. 
\end{proof}

\begin{lemma}
  Let $X$ be any iteration in $\sig T$.
  Then, the number of pairs an iteration $Y$ and an edge $e$ satisfying 
  $C(Y, e) = X$ is at most six.
      (A proof of Lemma~\ref{lem:corresponded_rec}. )
\end{lemma}

\begin{proof}
  Let $\sig L$ be the path from the root iteration $I$ on $\sig T$ to $X$ and
  $X'$ be the parent iteration of $X$. 
  Without loss of generality, 
  we can assume that $X$ is the top of a chain  
  by the definition of $C(Y, e)$. 
  Let $X''$ be the parent of the top of a chain including $X'$ and
  $\sig L'$ be the path from $X''$ to $X'$. 
  By the definition of $C(Y, e)$, 
  $Y$ satisfying $C(Y, e) = X$ exists only on $\sig L'$. 
  If $Y$ is $X'$, then the number of edges $e'$ satisfying 
  $C(Y, e') = X$ is at most two by conditions (1) and (2). 
  If $Y$ is $X''$, then the number of edges $e'$ satisfying 
  $C(Y, e') = X$ is at most two by conditions(3.b) and (4). 
  If $Y$ is not $X'$ and $X''$, 
  then the number of $Y$ satisfying $C(Y, e) = X$ is at most two from Lemma~\ref{lem:c3a}. 
  Hence, the statement holds. 
\end{proof}

%% file: appendix1.tex
\section{Appendix (The descriptive parameterized complexity of counting problems)}
\label{app:fo}


\newtheorem{fact}{Fact}

In this section, we will give the proof of Theorem~\ref{thm:im:fo} in Sec.~\ref{secmso}. See literatures~\cite{flum2006parameterized} 
There are some generic fixed-parameter tractability results for counting problems formulated in terms of descriptive complexity theory%
~\cite{flum2006parameterized}.
In 1991, Arnborg and Seese 
showed that the counting problems definable in monadic second-order logic (MSO) can be solved in linear time when the input structures have bounded tree-width~\cite{flum2006parameterized}. By extending these results, Frick~\cite{frick:TCSyst:2004generalized} showed that the counting problems definable in first-order logic (FO) can be solved in linear time on locally tree-decomposable classes of structures. 

Formally, their result is summarized as follows.
A \name{parameterized counting problem} is formulated as a problem of, given an instance $I$ and a parameterization $k\ge 0$, a parameterization of $I$, returning a number $F(I) \ge 0$. 
We say that a \name{counting problem} parameterized with $k$ belongs to \name{FPT} if there exist some computable function $f$ and constant $c > 0$ such that the problem can be solved in $O(f(k)\,n^c)$ time, where $n$ is the input size~\cite{flum2006parameterized}.

We introduce the language of first-order logic (FO) as follows~\cite{flum2006parameterized}. 
Let $\tau$ be the vocabulary of graphs.
For a graph $G = (V, E)$, we assume that the corresponding FO-structure $\sig A = (A, \dots)$ includes the domain $A := V\cup E$ for vertices and edges, and the incident relation $\Incd \subseteq E\times V\times V$ defined later.%
\footnote{
We can easily see that the structure $\sig A = (V\cup E, \Incd)$ is polynomially related to the standard graph structure $\sig A' = (V, E)$, and can be computed from $\sig A'$ in linear time. }
We also assume that the vocabulary $\tau$ for graphs includes the ternary relation $\Incd$ such that $\Incd(e, x, y)$ $\iff$ an edge $e$ has end points $x$ and~$y$. 
The set $FO[\tau]$ of \name{first-order formulas} is built up from
countably many sorted first-order variables $x, y, x_1, \dots$ for vertices and $e, f, e_1, \dots$ for edges, 
the predicate symbols $\Incd$ for incidencet relation, $=$ for equality, and $P_1, P_2, \dots \in \tau$, 
connectives $\lor, \land, \neg, \to$, and
quantifiers $\exists x, \forall x$ ranging over indivisuals.
We assume the standard semantics of FO-formulas~\cite{frick:TCSyst:2004generalized}. 

A \name{decision problem} $\Pi$ is \name{definable in FO} on vocabulary $\tau$%
for a class $\sig C$ of structures
\footnote{
This definition is equivalent to that a $\Pi$ is definable in FO on $\tau$ if there exists a FO-formula $\phi \in FO[\tau]$ such that for all structures  $\sig A$, the answer of $\Pi$ is yes on $G$ if and only if $\sig A \models \phi$.
}
if and only if there exists a FO-formula $\phi(\bar x) \in FO[\tau]$ with some $m$-vector of free variables $\bar x = (x_1,\dots,x_m)$ such that for all structure $\sig A \in \sig C$, the answer of $\Pi$ is yes on $\sig A$ $\iff$ there exists some $m$-tuple $\bar a = (a_1,\dots,a_m) \in A^m$ such that $\sig A \models \phi(\bar a)$ holds. Similarly, we introduce the FO-definability of counting problems as follows. 

\begin{definition}[\cite{frick:TCSyst:2004generalized,flum2006parameterized}]
A \name{counting problem} $\Pi$ is \name{definable in FO} if
there exists a FO-formula $\phi(\bar x) \in FO[\tau]$ with some $\bar x = (x_1,\dots,x_m)$ such that the answer (the number of solutions in $\Pi$) is given by the number of $\bar a \in A^m$ such that $\sig A \models \phi(\bar a)$.
\end{definition}




Extending the notion of tree-decomposition~\cite{flum2006parameterized}, Frick~\cite{frick:TCSyst:2004generalized} gave a sufficient condition, called locally tree-decomposability, so that a counting problems definable in the first-order logic (FO) belongs to FPT.
Let $\tau$ be a vocabulary and $\sig A$ a $\tau$-structure with domain $A$.
The \name{Gaifman graph} of $\sig A$ is the graph $\sig G(\sig A) = (A, F)$, where
the vertex set is $A$, and for every $a, b \in A$, $(a, b) \in F$ $\iff$ $a$ and $b$ appear in the same tuple of some relation in $\sig A$.
We denote by $d^{\sig A}(a, b)$ the distance between $a$ and $b$ in $\sig G(\sig A)$.
For any $r \ge 1$ and element $a \in A$, we define the \name{$r$-neighborhood} of $a$ by the set $N^{\sig A}_r(a) := \inset{ b \in A }{ d^{\sig A}(a, b) \le r }$.
For any set $X \subseteq A$, we define $N^{\sig A}_r(X) := \bigcup_{a \in X} N^{\sig A}_r(a)$.
  Let $\sig A$ be a structure with domain $A$. 
  Let $r, \ell \ge 1$ be any integers and $g : N \to N$ be any function.
  A \name{nice $(r, \ell, g)$-tree cover of a structure $\sig A$} is a class $\sig T$ of subsets of $A$ such that:
  \begin{enumerate}
  \item For any $a \in A$, there exists some $U \in \sig T$ such that
    $N^{\sig A}_r(a) \subseteq U$. 
    
  \item For any $U \in \sig T$, there are less than $\ell$ sets $V \in \sig T$ such that $U\cap V \not= \emptyset$. 
    
  \item For all $U_1, \dots, U_q \in \sig T$ and $q\ge 1$, 
    the tree-width of the induced structure $\langle U_1\cup\dots\cup U_q\rangle^{\sig A}$ is $g(q)$ or less. 
  \end{enumerate}

\begin{definition}[Frick~\cite{frick:TCSyst:2004generalized}]
  A class $\sig C$ of structures is \name{locally tree-decomposable} if there is a linear time algorithm that, given structure $\sig A \in \sig C$ and $r \ge 1$, computes a nice $(r, \ell, g)$-cover of $\sig A$
  for suitable $\ell = \ell(r)$ and $g = g_r$ depending only on $r$.
\end{definition}

In the above definition, the parameter $r$ corresponds to the radius of predicates in Gaifman's Theorem on the locality of FO-formulas~\cite{flum2006parameterized}. 

\begin{lemma}[Frick~\cite{frick:TCSyst:2004generalized}]
  \label{lem:frick:ldclass}
The classes of graphs of bounded degree, of bounded tree-width, and planar graphs are locally tree-decomposable. 
\end{lemma}

In the followings, we denote by $||G|| = m + n$ the size of an input graph $G$.


\begin{proposition}[Frick~\cite{frick:TCSyst:2004generalized}]
  \label{prop:fpt:fo:frick}
  Let $\sig C$ be any class of locally tree-decomposable structures.
  For any structure $\sig{A} \in \sig C$, a counting problem $\Pi$ definable in FO can be solved in linear time in $||\sig{A}||$,
  where $\sig{A}$ is given with its underlying nice $(r, \ell, g)$-tree cover $\sig T$ associated with $r, \ell, g$. 
\end{proposition}

From the proof of Proposition~\ref{prop:fpt:fo:frick} in~\cite{frick:TCSyst:2004generalized}, we can see that the running time of the algorithm to solve the the counting problem is $O(f(k) n^c)$ for some function $f$ and constant $c$, where the parameter $k$ is determined by on $r, \ell, g$, the nice $(r, \ell, g)$-tree cover $\sig T$, and the size $||\phi||$ of the formula $\phi$.
Note that $f$ is even not elementary in general~\cite{frick:TCSyst:2004generalized}.

For any $k\ge 0$, a \name{$k$-induced matching} in a graph $G = (V, E)$ is an induced matching $M \subseteq E$ with $|M| = k$.
Now, we show the main results of this section. 

\begin{theorem}
  For any class $\sig G$ of graphs of bounded degree, graphs of bounded tree-width, or planar graphs and any $k\ge 0$, the counting problem of $k$-induced matchings in an input graph $G$ in $\sig G$ can be solved in linear time in $||G||$.
    (A proof of Theorem~\ref{thm:im:fo}. )
\end{theorem}

\begin{proof}
  To prove the theorem, it is sufficient to give a FO-formula $\phi(e_1,\dots,e_k)$ such that
  $G \models \phi(e_1,\dots,e_k)$ $\iff$ the set $M = \set{e_1,\dots,e_k}$ is an induced matching in $G$.
  We give auxiliary predicates as follows.
First, the predicate 
\begin{eqnarray*}
  \textit{DISJ}(e_1,\dots,e_k)
  &\equiv&
  \bigwedge_{1\le i < j \le k}
    \neg (e_i = e_j)
\end{eqnarray*}
states that edges $e_1,\dots,e_k$ are mutually distinct.
Secondly, the predicate 
\begin{eqnarray*}
  \textit{CONN}(e, f, g)
  &\equiv&
  \exists x_1, \exists x_2, \exists y_1, \exists y_2, \exists z_1, \exists z_2\; \\
  &&\left\{\begin{array}{l}
  \Incd(e, x_1, x_2) \land \Incd(f, y_1, y_2) \land \\
  \left[\begin{array}{l}
    (\Incd(g, z_1, z_2)  
    \land\; \bigvee_{i, j \in\set{1,2}} (x_i = z_1) \land (y_j = z_2) \; )
    \\
    \lor (x_1 = y_1)
    \lor (x_2 = y_2)
  \end{array}\right]
  \end{array}\right\}.
\end{eqnarray*}
states that given edges $e$ and $f$ are connected by another given edge $g$.
Using $\textit{CONN}$, we then define the predicate 
\begin{eqnarray*}
  \textit{SAFE}(e_1,\dots,e_k)
  &\equiv&
  \bigwedge_{1\le i < j \le k}
  \neg 
  \left\{\begin{array}{c}
  \exists g\, \textit{CONN}(e_i, e_j, g)
  \end{array}\right\}
\end{eqnarray*}
which states that there is no edge connecting any pair of distinct edges in $\set{e_1,\dots,e_k}$. Combining these predicates, we build the desired predicate
\begin{eqnarray*}
  \phi(e_1,\dots,e_k)
  &\equiv&
  \textit{DISJ}(e_1,\dots,e_k)
  \land
  \textit{SAFE}(e_1,\dots,e_k)
\end{eqnarray*}
Hence, it follows from~Lemma~\ref{lem:frick:ldclass} and Proposition~\ref{prop:fpt:fo:frick} that for any $G = (V, E)$, we can compute the number $N_k$ of $k$-tuples $\bar e = (e_1,\dots,e_k) \in E^k$ with $G \models \phi(\bar e)$ in linear time in $||G|| = m + n$. 
Since there are exactly $k!$ permutations of $\set{e_1,\dots,e_k}$, we can obtain the number of $k$-induced matchings in $G$ as $N_k$ devided by $k!$. 
\end{proof}

From the proof of the above theorem, we see that the counting problem of $k$-induced matchings belongs to FPT when parameterized with $k$ and some constants determined by the nice local tree decomposition for an input class $\sig G$.
It is not hard to see that the formula $\phi$ in the proof can be written in the form $\phi(\bar e) \equiv \forall \bar x\,\psi(\bar e, \bar x)$ for a quantifier-free FO-formula $\psi$. Thus, it is in FO-$\Pi_1$.